\documentclass[journal]{IEEEtran}
\usepackage{graphics} 
\usepackage{epsfig} 
\usepackage{amsmath} 
\usepackage{amsthm}
\usepackage{amssymb}  

\usepackage{mathrsfs}

\usepackage{subfigure}
\usepackage{color} 
\usepackage{cite}
\usepackage{mathtools}
\usepackage{mathrsfs}
\usepackage{txfonts}
\usepackage{amssymb}
\usepackage{multirow}
\usepackage{subfigure}

\newtheorem{assumption}{Assumption}[section]

\newcommand{\HEI}[1]{\bf Hybrid-EI}

\newtheorem{theorem}{Theorem}
\newtheorem{remark}{Remark}
\newtheorem{example}{Example}
\newtheorem{definition}{Definition}

\ifCLASSINFOpdf

\else

\fi

\begin{document}
%

\title{Event-Triggered Control for Nonlinear \\ Time-Delay Systems}

%
%
%

\author{Kexue~Zhang~~~~
        ~Bahman~Gharesifard~~~~
        ~Elena~Braverman
\thanks{This work was supported by the Natural Sciences and Engineering Research Council of Canada (NSERC), the grant RGPIN-2020-03934. The first author was supported by a fellowship from the Pacific Institute of the Mathematical Sciences (PIMS), Canada.}
\thanks{K. Zhang and E. Braverman are with the Department
of Mathematics and Statistics, University of Calgary, Calgary, Alberta T2N 1N4, Canada (e-mails: \{kexue.zhang,~maelena\}@ucalgary.ca).

B. Gharesifard is with the Department
of Mathematics and Statistics, Queen's University, Kingston, Ontario K7L 3N6, Canada (e-mail: bahman.gharesifard@queensu.ca).}
}

\maketitle
%
\begin{abstract}
This paper studies the event-triggered control problem of general nonlinear systems with time delay. A novel event-triggering scheme is presented with two tunable design parameters, based on a Lyapunov functional result for input-to-state stability of time-delay systems. The proposed event-triggered control algorithm guarantees the resulting closed-loop systems to be globally asymptotically stable, uniformly bounded and/or globally attractive for different choices of these parameters. Sufficient conditions on the parameters are derived to exclude Zeno behavior. Two illustrative examples are studied to demonstrate our theoretical results.
\end{abstract}
%
%
%
%
%
%
%

\section{Introduction}\label{Sec1}
%
%
%
%
%
%
%
%

\IEEEPARstart{E}{vent}-triggered control provides an effective way to update the control signals at a sequence of discrete-time moments determined by certain execution rule, often referred to as an \emph{event}. The main advantage is to improve the efficiency of control implementations while still guaranteeing the desired performance levels of the closed-loop systems. Since the flagship work~\cite{PT:2007} was published more than a decade ago, the control community has shown an increasing interest in the study of event-triggering algorithms and their corresponding applications. Up to now, the applications of event-triggered control have been found in a wide variety of control problems, notably in consensus and synchronization problems, distributed optimization, economic dispatch, robot operation, and vehicle platooning (see, e.g., the survey papers \cite{ZPJ-TFL:2015,CN-EG-JC:2019,QL-ZW-XH-DHZ:2014} and references therein).

Time delay exists widely in nature and frequently occurs in many practical systems (see monographs \cite{EF:2014,MK:2009} and many references therein). Evolution of a time-delay system, typically modeled by functional differential equations (see, e.g., \cite{JKH:1977}), depends not only on the present states but also the states at some historical moments. Time-delay effects have been considered with the event-triggered \emph{controllers} for numerous systems that themselves are delay-free. For example, event-based consensus of multi-agent systems was studied in \cite{WZ-ZPJ:2015,EG-YC-DWC:2017} with time delay effects considered in the consensus protocols. To compensate for the delay in the control inputs, the predictor-based control method recently has been combined with the event-triggered control mechanism for stabilization of control systems (see, e.g., \cite{AS-EF:2016,EN-PT-JC:2020}). However, the study of event-triggered control for time-delay systems has just started to gain attentions in the past few years. The main topic of this research is to provide a framework for incorporating time delays in the control systems with the design of event-triggered controllers, while ensuring the absence of \emph{Zeno behavior}, i.e., making an infinite number of control updates in a finite time. For systems without time delay, one of the main ideas for ruling out Zeno behavior, initially introduced in \cite{PT:2007}, is to analyze the dynamics of ${\|\epsilon\|}/{\|x\|}$, where $\epsilon$ represents a certain measurement error, and $x$ is the system state. However, such method for Zeno behavior exclusion cannot be extended seamlessly to time-delay systems mainly because the state can be zero in a finite time due to the existence of time delays. Since it is important for the purpose of control implementations to rule out Zeno behavior with the designed event-triggered controller, alternative approaches are crucially needed for time-delay systems.

In the past few years, a handful of results have been reported for event-triggered control of time-delay systems, particularly in the area of network control systems. A commonly used method to exclude Zeno behavior is based on hybrid event-triggering strategies. For instance, a hybrid event-triggering scheme, which considers the switching between periodic state sampling and continuous event-triggering, is proposed in \cite{ZF-CG-HG:2017}. The natural feature of this scheme is the automatic avoidance of Zeno behavior due to state sampling. In \cite{KZ-BG:2019}, a hybrid impulsive and event-triggered control algorithm is constructed for nonlinear time-delay systems. The impulsive controller plays an important role in ruling out Zeno behavior with a prescribed lower bound of the inter-execution times. So far, several event-triggering control schemes have been successfully developed for some particular time-delay systems. Event-based stabilization of nonlinear time-delay systems was studied in~\cite{SD-NM-JFG:2014}, where the considered systems are affine in the control, and the proposed event-triggering rule and the feedback controller both require the full knowledge of the system delays, and such requirement is essential in ruling out Zeno behavior. Event-triggered stabilization of stochastic nonlinear systems with discrete delays and external disturbances was initially investigated in~\cite{QZ:2019} in which the global Lipschitz condition is required on the nonlinear dynamics. There are few other interesting event-triggering schemes introduced for some particular network systems with sufficient conditions provided for the exclusion of Zeno behavior (see, e.g., \cite{YD-JWC-JGX:2018,YC-SW-ZG-TH-SW:2019,JC-BC-ZZ-PJ:2020} for more details). Recently, event-triggering algorithms under sample-and-hold implementations were constructed in\cite{AB-PP-IDL-MDF:2021,AB-PP:2020} for stabilization of general nonlinear systems with time delay. The designed triggering conditions are only examined at a sequence of state-sampling instants so that a minimum inter-event time can be naturally guaranteed. However, the knowledge of the maximum involved delay and a memory of the system states at some historic moments are required for each control update. It can be seen that the study of event-triggering designs with guaranteed Zeno behavior exclusion for general nonlinear time-delay systems has not been extensively conducted, and we set this as our main objective in this research. The main contributions of this paper are described as follows. 


\textit{Statement of Contributions.} Based on a Lyapunov functional result for input-to-state stability of time-delay systems, we propose a novel event-triggered control algorithm for stabilization of general nonlinear systems with time delay. This algorithm invokes the control updates when a certain measurement error reaches a dynamic threshold which is a sum of two functions: one is a function of the system states and another is a function of the evolution time. 
\begin{itemize}
\item Without the time-dependent function, the state-dependent function in our algorithm guarantees the asymptotic stability of the closed-loop systems, but using this alone, Zeno behavior cannot be avoided. 
\item With the time-dependent function, we show that we can exclude Zeno behaviour while also assuring boundedness and attractivity of the trajectories; however, we show that with this strategy, one cannot in general guarantee that the trajectories stay arbitrarily close to the equilibrium by setting the initial data close enough to the equilibrium, i.e., stability may not be maintained. 
\end{itemize}
In the sense described above, our results demonstrate an interesting interplay between establishing stability properties and excluding Zeno behaviour. It is also worthwhile to mention that our algorithm preserves the performance of trajectory boundedness and attractivity which naturally leads to applications in various control problems, e.g., synchronization of dynamical networks, consensus of multi-agent systems, and distributed optimization. Compared with the most recent work on event-triggered control of general nonlinear time-delay systems in \cite{AB-PP-IDL-MDF:2021,AB-PP:2020}, our results are delay-independent, that is, the exact value of the maximum involved delay in the system is not necessary for the computation of control updates, and the proposed event-triggering condition does not require the memory of the delayed states.


The rest of this paper is organized as follows. Some mathematical preliminaries and a Lyapunov functional result for input-to-state stability of time-delay systems are provided in Section \ref{Sec2}. We propose an event-triggered control algorithm with the main results in Section \ref{Sec3}.  Illustrative examples are presented with some numerical simulations in Section \ref{Sec4} to demonstrate the effectiveness of the obtained theoretical results. Section \ref{Sec5} concludes this paper and outlines some future research directions.

\section{Preliminaries}\label{Sec2}
Let $\mathbb{N}$ denote the set of positive integers, $\mathbb{R}$ the set of real numbers, $\mathbb{R}^+$ the set of nonnegative reals, and $\mathbb{R}^n$ the $n$-dimensional real space equipped with the Euclidean norm denoted by $\|\cdot\|$. For $a,b\in \mathbb{R}$ with $b>a$, define
\begin{align*}
\mathcal{PC}([a,b],\mathbb{R}^n)=& \{ \varphi:[a,b]\rightarrow\mathbb{R}^n \mid \varphi \textrm{ is piecewise} \textrm{  right-}\cr
& \textrm{continuous} \}\cr
\mathcal{PC}([a,\infty),\mathbb{R}^n) =&  \{\phi:[a,\infty)\rightarrow\mathbb{R}^n \mid \phi|_{[a,c]}\in \mathcal{PC}([a,c],\mathbb{R}^n) \cr
 & \textrm{ for all } c>a \}
\end{align*}
where $\phi|_{[a,c]}$ is a restriction of $\phi$ on interval $[a,c]$. Let $\mathcal{C}(J,\mathbb{R}^n)$ denote the set of continuous functions mapping interval $J$ to $\mathbb{R}^n$. Given $\tau>0$, the linear space $\mathcal{C}([-\tau,0],\mathbb{R}^n)$ is equipped with a norm defined by $\|\varphi\|_{\tau}:=\sup_{s\in[-\tau,0]}\|\varphi(s)\|$ for $\varphi\in \mathcal{C}([-\tau,0],\mathbb{R}^n)$. For the sake of simplicity, we use $\mathcal{C}_{\tau}$ to represent $\mathcal{C}([-\tau,0],\mathbb{R}^n)$.

Consider the time-delay control system:
\begin{eqnarray}\label{sys}
\left\{\begin{array}{ll}
\dot{x}(t)=f(t,x_t,u) \cr
x_{t_0}=\varphi
\end{array}\right.
\end{eqnarray}
where $x(t)\in\mathbb{R}^n$ denotes the system state at time $t$; $u\in \mathcal{PC}([t_0,\infty),\mathbb{R}^m)$ represents the input; $\varphi\in\mathcal{C}_\tau$ is the initial function; $f:\mathbb{R}^+\times\mathcal{C}_{\tau}\times\mathbb{R}^m\rightarrow\mathbb{R}^n$ satisfies $f(t,0,0)=0$ for all $t\in\mathbb{R}^+$, which implies that system~\eqref{sys} without the input $u$ admits the zero solution (trivial solution); $x_{t}$ is defined as $x_{t}(s):=x(t+s)$ for $s\in[-\tau,0]$, and $\tau>0$ is the maximum involved delay. Given $u\in \mathcal{PC}([t_0,\infty),\mathbb{R}^m)$, define $F(t,\phi):=f(t,\phi,u(t))$ and suppose $F$ satisfies the necessary conditions\footnote{Since $u\in \mathcal{PC}([t_0,\infty),\mathbb{R}^m)$, it is possible for the function $F$ to be discontinuous at the discontinuity points of $u$. Hence, the fundamental theory introduced in \cite{GB-XL:1999} for impulsive time-delay systems applies to system~\eqref{sys} with discontinuous function $F$.} in \cite{GB-XL:1999} so that, for any initial function $\varphi\in\mathcal{C}_{\tau}$, system \eqref{sys} has a unique solution $x(t,t_0,\varphi)$ that exists in a maximal interval $[t_0-\tau,t_0+\Gamma)$, where $0<\Gamma\leq \infty$.

The notion of input-to-state stability (ISS), introduced by Sontag in~\cite{EDS:1989}, has been proved powerful to characterize the effects of external inputs, and recently has shown great application potential in the design of event-triggered controllers (see, e.g. \cite{WZ-ZPJ:2015,PT:2007}). Our event-triggered control algorithm relies heavily on this notion. We recall the following function classes before giving the formal ISS definition for system \eqref{sys}. A continuous function $\alpha:\mathbb{R}^+\rightarrow\mathbb{R}$ is said to be of class $\mathcal{K}$ and we write $\alpha\in\mathcal{K}$, if $\alpha$ is strictly increasing and $\alpha(0)=0$. If $\alpha\in\mathcal{K}$ and also $\alpha(s)\rightarrow\infty$ as $s\rightarrow\infty$, we say that $\alpha$ is of class $\mathcal{K}_{\infty}$ and we write $\alpha\in\mathcal{K}_{\infty}$. A continuous function $\beta:\mathbb{R}^+\times\mathbb{R}^+ \rightarrow\mathbb{R}^+$ is said to be of class $\mathcal{KL}$ and we write $\beta\in \mathcal{KL}$, if the function $\beta(\cdot,t)\in\mathcal{K}$ for each fixed $t\in\mathbb{R}^+$, and the function $\beta(s,\cdot)$ is decreasing and $\beta(s,t)\rightarrow 0$ as $t\rightarrow \infty$ for each fixed $s\in \mathbb{R}^+$.

Now we are ready to state the ISS definition for system \eqref{sys}.

\begin{definition}
System \eqref{sys} is said to be {input-to-state {stable}} (ISS) if there exist functions $\beta\in\mathcal{KL}$ and $\gamma\in\mathcal{K}$ such that, for each initial function $\varphi\in\mathcal{C}_{\tau}$ and input function $u\in\mathcal{PC}([t_0,\infty),\mathbb{R}^m)$, the corresponding solution to \eqref{sys} exists globally and satisfies
$$\|x(t)\|\leq \beta\left(\|\varphi\|_{\tau},t-t_0\right)+\gamma\left(\sup_{s\in[t_0,t]}\|u(s)\|\right)~ \mathrm{~for~all~} t\geq t_0.$$
\end{definition}
Next, we present several concepts regarding Lyapunov candidates and review a Lyapunov functional result which will be used for the design of our event-triggering algorithm. A function $V:\mathbb{R}^+\times\mathbb{R}^n\rightarrow \mathbb{R}^+$ is said to be of class $\mathcal{V}_0$ and we write $V\in\mathcal{V}_0$, if, for each $x\in\mathcal{C}(\mathbb{R}^+,\mathbb{R}^n)$, the composite function $t\mapsto V(t,x(t))$ is in $\mathcal{C}(\mathbb{R}^+,\mathbb{R}^+)$. A functional $V:\mathbb{R}^+\times \mathcal{C}_{\tau}\rightarrow\mathbb{R}^+$ is said to be of class $\mathcal{V}^*_0$ and we write $V\in \mathcal{V}^*_0$, if, for each function $x\in\mathcal{C}([-\tau,\infty),\mathbb{R}^n)$, {the composite function $t\mapsto V(t,x_t)$ is continuous in $t$ for all $t\geq 0$}, and $V$ is locally Lipschitz in its second argument (such Lipschitz condition will be introduced formally in Definition \ref{Lipschiz}). Given an input $u\in \mathcal{PC}([t_0,\infty),\mathbb{R}^m)$, we define the upper right-hand derivative of the Lyapunov functional candidate $V(t,x_t)$ with respect to system \eqref{sys}:
\[
\mathrm{D}^+V(t,\phi)=\limsup_{h\rightarrow 0^+}\frac{V\left(t+h,x_{t+h}(t,\phi)\right)-V(t,\phi)}{h}
\]
where $x(t,\phi)$ is a solution to \eqref{sys} satisfying $x_t=\phi$, that is, $x(s):=x(t,\phi)(s)$ is a solution to the initial value problem: 
\[
\left\{\begin{array}{ll}
\dot{x}(s)=f(s,x_s,u(s))\cr
x_{s_0}=\phi
\end{array}\right.
\]
for $s\in[t,t+h)$ and $s_0=t$, where $h$ is a small positive number.

Reference \cite{XL-KZ:2019} studied a more general form of system~\eqref{sys} with impulse effects (state abrupt changes or jumps). Here, we review a special case in which no impulses are considered and the corresponding ISS result is as follows (see also~\cite{PP-ZPJ:2006}).

\begin{theorem}\label{Th.ISS}
Assume that there exist functions $V_1\in \mathcal{V}_0$, $V_2\in \mathcal{V}^*_0$, $\alpha_1, \alpha_2, \alpha_3 \in \mathcal{K}_{\infty}$ and $\chi\in \mathcal{K}$, and constant $\mu>0$ such that, for all $t\in \mathbb{R}^+$ and $\phi\in \mathcal{C}_{\tau}$, 
\begin{itemize}
\item[(i)] $\alpha_1(\|\phi(0)\|)\leq V_1(t,\phi(0))\leq \alpha_2(\|\phi(0)\|)$; 

\item[(ii)] $0\leq V_2(t,\phi)\leq \alpha_3(\|\phi\|_{\tau})$;

\item[(iii)] $V(t,\phi):=V_1(t,\phi(0))+V_2(t,\phi)$ satisfies
\[
\mathrm{D}^+V(t,\phi) \leq -\mu V(t,\phi) +\chi(\|u\|).
\]

\end{itemize}
Then system \eqref{sys} is ISS.
\end{theorem}

It can be seen from the proof of Theorem \ref{Th.ISS} in \cite{XL-KZ:2019} that the global asymptotic stability (GAS) of system \eqref{sys} is guaranteed when $u=0$. Let $\alpha_4:=\alpha_2+\alpha_3$, and note that $\alpha_4\in\mathcal{K}_{\infty}$ and conditions~(i) and~(ii) of Theorem~\ref{Th.ISS} imply
\[
\alpha_1(\|\phi(0)\|)\leq V(t,\phi)\leq \alpha_4(\|\phi\|_{\tau}),
\] 
which is a standard condition on the Lyapunov candidate when using the method of Lyapunov functionals to justify the stability of a time-delay system (see, e.g., \cite{XL-KZ:2019,PP-ZPJ:2006}). The decomposition of $V$ into a function portion $V_1$ and a functional portion $V_2$ was introduced for impulsive systems with time delay (see, e.g.,  \cite{XL-KZ:2019}). The function $V_1$ is used to analyze the impulse effects on the entire Lyapunov candidate $V$, whereas the functional $V_2$ is indifferent to impulses. Although we will not consider the impulse effects in system~\eqref{sys}, such decomposition with condition (iii) of Theorem \ref{Th.ISS} will play an essential role in the event-triggered controller design, as demonstrated in the next section.

\section{Event-Triggered Control Algorithm}\label{Sec3}
Consider feedback control system \eqref{sys} with a sampled-data implementation:
\begin{eqnarray}\label{ETC.sys}
\left\{\begin{array}{ll}
\dot{x}(t)=f(t,x_t,u(t)) \cr
u(t)=k(x(t_i)),{~t\in[t_{i},t_{i+1})}\cr
x_{t_0}=\varphi
\end{array}\right.
\end{eqnarray}
where $u:\mathbb{R}^+\rightarrow \mathbb{R}^m$ is a control input and $k: \mathbb{R}^n\rightarrow \mathbb{R}^m$ is the feedback control law and satisfies $k(0)=0$. Therefore, system~\eqref{ETC.sys} admits the trivial solution. The time sequence $\{t_i\}_{i\in\mathbb{N}}$ is to be determined according to certain execution rule defined later based on the state measurement, and each time instant $t_i$ corresponds to a control update $u(t_i)$. To be more specific, the controller receives the states, calculates the control law, and updates the control signals all at time $t=t_i$ while staying unchanged between consecutive control updates.

Define the measurement error of state by
\begin{equation}\label{error}
\epsilon(t)=x(t_i)-x(t)~ \textrm{ for } t\in [t_{i},t_{i+1}),
\end{equation}
and note that $\epsilon(t_{i})=0$ for all $i\in\mathbb{N}$ and $\epsilon\in \mathcal{PC}([t_0,\infty),\mathbb{R}^n)$. The feedback control $u$ can then be written as follows
\begin{equation}\label{ETCer}
u(t)=k(x(t_i))=k(\epsilon(t)+x(t))~ \textrm{ for } t\in [t_i,t_{i+1}).
\end{equation}
We thus derive the following closed-loop system by substituting \eqref{ETCer} into system \eqref{ETC.sys}: 
\begin{eqnarray}\label{CL.sys}
\left\{\begin{array}{ll}
\dot{x}(t)=f(t,x_t,k(\epsilon+x)) \cr
x_{t_0}=\varphi.
\end{array}\right.
\end{eqnarray}
Throughout this paper, we make the following assumption on control system~\eqref{CL.sys}.
\begin{assumption}\label{A1}
There exist functions $V_1\in \mathcal{V}_0$, $V_2\in \mathcal{V}^*_0$, $\alpha_1, \alpha_2, \alpha_3 \in \mathcal{K}_{\infty}$ and $\chi\in \mathcal{K}$, and constant $\mu>0$ such that all the conditions of Theorem \ref{Th.ISS} hold for system \eqref{CL.sys} with input $u$ replaced by the measurement error $\epsilon$. 
\end{assumption}
With the above assumption, the closed-loop system \eqref{CL.sys} is ISS with respect to the measurement error $\epsilon$, and GAS when $\epsilon=0$. Boundedness of the state follows from Assumption~\ref{A1}, and then the global existence of the unique solution to system~\eqref{ETC.sys} is guaranteed (see \cite{GB-XL:1999}).

For the completion purpose, we recall the following two definitions related to system~\eqref{CL.sys} (or equivalently, system~\eqref{ETC.sys}). 
\begin{definition}[Global Uniform Boundedness] System~\eqref{CL.sys} is said to be globally uniformly bounded (GUB), if for any $\delta>0$ there exists a constant $M:=M(\delta)>0$ such that $\|\varphi\|_{\tau}\leq \delta \Rightarrow \|x(t)\|\leq M$ for all $t\geq t_0$, where $x(t):=x(t,t_0,\varphi)$ is the solution of system~\eqref{CL.sys}.
\end{definition}

\begin{definition}[Global Attractivity] The {trivial solution} of system~\eqref{CL.sys} is said to be globally attractive (GA), if $\lim_{t\rightarrow\infty} \|x(t)\|=0$ for any $\varphi\in\mathcal{C}_{\tau}$, where $x(t):=x(t,t_0,\varphi)$ is {the solution} of~\eqref{CL.sys}.
\end{definition}

In this paper, we propose an execution rule to determine the time sequence $\{t_i\}_{i\in\mathbb{N}}$ for the updates of the feedback control input $u$ so that the closed-loop system \eqref{CL.sys} with the measurement error $\epsilon$ still preserves certain desired performance (such as GAS, boundedness, and attractivity). To do so, we restrict $\epsilon$ to satisfy
\begin{equation}\label{exrule}
\chi(\|\epsilon\|)\leq \sigma \alpha_1(\|x\|) + \chi\left(a e^{-b(t-t_0)}\right) 
\end{equation}
for some constants $\sigma\geq 0$, $a\geq 0$, and $b>0$; the updating of the control input $u$ is then triggered by the execution rule (or event)
\begin{equation}\label{event}
\chi(\|\epsilon\|) = \sigma \alpha_1(\|x\|) + \chi\left(a e^{-b(t-t_0)}\right).
\end{equation}
The event times are the moments when the event occurs, i.e., 
\begin{equation}\label{et.time}
t_{i+1}=\inf\left\{t\geq t_i \mid \chi(\|\epsilon\|) = \sigma \alpha_1(\|x\|)+ \chi\left(a e^{-b(t-t_0)}\right)\right\}.
\end{equation}
According to the feedback control law in \eqref{ETC.sys}, the control input is updated at each $t_i$ (the error $\epsilon$ is set to zero simultaneously), remains constant until the next event time $t_{i+1}$, and then the error $\epsilon$ is reset to zero again. The event times in \eqref{et.time} are defined implicitly, and then it is necessary to rule out the existence of Zeno behavior, which is formalized in the following definition. 
\begin{definition}[Zeno Behavior\cite{CN-EG-JC:2019}]\label{Zeno}
For the closed-loop system~\eqref{CL.sys} with event times determined by~\eqref{et.time}, a solution with initial condition $x_{t_0}=\varphi$ is said to exhibit Zeno behavior, if there exists $T>0$ such that $t_i\leq T$ for all $i\in \mathbb{N}$. If Zeno behavior does not occur along any solution of system~\eqref{CL.sys}, then we say system~\eqref{CL.sys} does not exhibit Zeno behavior.
\end{definition}

The objective for the rest of this paper is to establish sufficient conditions to guarantee the desired performance of system~\eqref{CL.sys} with event times determined by~\eqref{et.time} and also exclude Zeno behavior. In general, there are two main ways that Zeno behavior can be ruled out: the first one is to ensure the existence of a uniform lower bound of any inter-execution time, and the second one is based on the contradiction argument to show the exclusion of Zeno behavior by a direct use of Definition~\ref{Zeno}. We will derive our results by using these two methods. It is worth mentioning that ensuring a lower bound of the inter-execution times is stronger than ruling out Zeno behavior. For example, suppose system~\eqref{CL.sys} has a solution with event times $t_k=\sum^k_{i=1}1/i$ for $k\in\mathbb{N}$. It can be seen that $t_k\rightarrow\infty$ as $k\rightarrow\infty$ which means such a solution does not exhibit Zeno behavior. However, a lower bound of the inter-execution times is not ensured, since $t_{k+1}-t_k\rightarrow 0$ as $k\rightarrow\infty$. We refer the reader to~\cite{CN-EG-JC:2019} for a detailed discussion on Definition~\ref{Zeno}.

In this study, we mainly require locally Lipschitz conditions on the functions $\alpha_1^{-1}$ (the inverse function of $\alpha_1$), $\chi$, $k$, and $f$, which are described as follows.
\begin{definition}[Locally Lipschitz]\label{Lipschiz} In this paper, local Lipschitz comes in two different flavors.
\begin{itemize}
\item The function $f:\mathbb{R}^n\rightarrow\mathbb{R}^m$ is called locally Lipschitz, if for each $z\in \mathbb{R}^n$ there exist constants $L>0$ and $R>0$ such that $\|f(y)-f(z)\| \leq L \|y-z\|$ on the open ball of center $z$ and radius $R$:
\[
\mathcal{B}_{R}(z):= \{y\in\mathbb{R}^n \mid \|y-z\| < R  \}.
\]

\item The function $f:\mathcal{C}_{\tau}\rightarrow\mathbb{R}^n$ is called locally Lipschitz, if for each $\phi\in \mathcal{C}_{\tau}$ there exist positive constants $L$ and $R$ such that $\|f(\varphi)-f(\phi)\| \leq L \|\varphi-\phi\|_{\tau}$ on the open ball of center $\phi$ and radius $R$:
\[
\mathcal{B}^{\tau}_{R}(\phi):=\{\varphi\in\mathcal{C}_{\tau} \mid \|\varphi-\phi\|_{\tau} < R  \}.
\]
\end{itemize}
\end{definition}

Now we are in the position to introduce our first result.

\begin{theorem}\label{Th}
Suppose that Assumption \ref{A1} holds with $V_1\in \mathcal{V}_0$, $V_2\in \mathcal{V}^*_0$, $\alpha_1, \alpha_2, \alpha_3 \in \mathcal{K}_{\infty}$ and $\chi\in \mathcal{K}$, and constant $\mu>0$. The event times $\{t_i\}_{i\in\mathbb{N}}$ are defined by~\eqref{et.time} with $0\leq \sigma<\mu$, $a\geq 0$, and $b>0$. We further assume that
\begin{itemize}
\item $\alpha^{-1}_1$, $\chi$, and $k$ are locally Lipschitz;

\item $f(t,\phi,u)$ is locally Lipschitz in $\phi$ for all $(t,u)\in \mathbb{R}^+\times \mathbb{R}^m$;

\item $f(t,\phi,u)$ is locally Lipschitz in $u$ for all $(t,\phi)\in \mathbb{R}^+\times \mathcal{C}_{\tau}$
\end{itemize}
then,
\begin{itemize}
\item[1)] if $a=0$ and $\sigma>0$, the closed-loop system~\eqref{CL.sys} is GAS;

\item[2)] if $a>0$, the closed-loop system~\eqref{CL.sys} is GUB and its trivial solution is GA. Moreover, if $b<\mu-\sigma$, then the inter-execution times $\{t_{i+1}-t_i\}_{i\in\mathbb{N}}$ are lower bounded by a quantity $T^*>0$, that is, $t_{i+1}-t_i \geq T^*$ for any $i\in\mathbb{N}$.
\end{itemize}
\end{theorem}

\begin{proof} For the sake of notational convenience, let $v_1(t):=V_1(t,x)$, $v_2(t):=V_2(t,x_t)$, and $v(t):=V(t,x_t)=v_1(t)+v_2(t)$. Execution rule~\eqref{event} guarantees that
\begin{align*}
\chi(\|\epsilon\|) &\leq \sigma \alpha_1(\|x\|) + \chi\left(a e^{-b(t-t_0)}\right)\cr
            &\leq \sigma \alpha_1(\|x\|) + aL e^{-b(t-t_0)}
\end{align*}
for all $t\geq t_0$, where $L>0$ denotes the Lipschitz constant of $\chi$ on the interval $[0,a]$. We then can derive from condition (ii) of Theorem~\ref{Th.ISS} that
\begin{align}\label{dini}
\mathrm{D}^+v(t) &\leq -\mu v(t) +\chi\left(\|\epsilon\|\right) \cr
                      &\leq -\mu v(t) +\sigma \alpha_1(\|x\|) + aL e^{-b(t-t_0)}\cr
                      &\leq -c v(t) + aL e^{-b(t-t_0)}
\end{align}
where $c:=\mu-\sigma>0$. We derived the third inequality of~\eqref{dini} from condition (i) of Theorem~\ref{Th.ISS} and the fact that $v_1(t)\leq v(t)$ for all $t\geq t_0$. Multiply both sides of~\eqref{dini} by $e^{c(t-t_0)}$ and then we have 
\begin{equation}\label{dini1}
\mathrm{D}^+ [e^{c(t-t_0)}v(t) ]\leq aL e^{(c-b)(t-t_0)}.
\end{equation}
Integrating both sides of~\eqref{dini1} from $t_0$ to $t> t_0$ yields
\begin{equation}
e^{c(t-t_0)} v(t)-v(t_0) \leq aL \int^t_{t_0} e^{(c-b)(s-t_0)} \mathrm{d} s .
\end{equation}
If $b>c$, we have
\begin{align}
v(t) &\leq v(t_0)e^{-c(t-t_0)} + \frac{aL}{b-c}  (e^{-c(t-t_0)}-e^{-b(t-t_0)} )\cr
     &\leq e^{-c(t-t_0)} \left(v(t_0) +\frac{aL}{b-c}  (1-e^{-(b-c)(t-t_0)} ) \right)\cr
     &\leq e^{-c(t-t_0)} \left(v(t_0) +\frac{aL}{b-c}  \right).
\end{align}
If $b<c$, we then get
\begin{align}
v(t) &\leq e^{-b(t-t_0)} \left(v(t_0)e^{-(c-b)(t-t_0)} +\frac{aL}{c-b}  (1-e^{-(c-b)(t-t_0)} ) \right)\cr
     &\leq e^{-b(t-t_0)} \left(v(t_0) +\frac{aL}{c-b}  \right).
\end{align}
If $b=c$, then
\begin{align}\label{b=c}
v(t) &\leq v(t_0)e^{-c(t-t_0)} + aL(t-t_0)e^{-c(t-t_0)}\cr
     &\leq v(t_0)e^{-c(t-t_0)} + \frac{aL}{c-\xi}e^{-\xi(t-t_0)}\cr
     &\leq e^{-\xi(t-t_0)} \left( v(t_0) + \frac{aL}{c-\xi} \right)     
\end{align}
where constant $\xi$ satisfies $0<\xi<c$. To derive the second inequality of~\eqref{b=c}, we applied the fact that $(c-\xi)(t-t_0)\leq e^{(c-\xi)(t-t_0)}$ for $t\geq t_0$. Define
\begin{align}\label{bound}
\eta:=\left\{\begin{array}{ll}
\min\{b,c\},~&\textrm{if}~ b\not=c \cr
\xi,~&\textrm{if}~ b=c
\end{array}\right.
\end{align}
and
\[
\bar{M}:=\left\{\begin{array}{ll}
\frac{aL}{|c-b|},~&\textrm{if}~ b\not=c \cr
\frac{aL}{|c-\xi|},~&\textrm{if}~ b=c
\end{array}\right.
\]
then
\begin{equation}
v(t)\leq (v(t_0)+\bar{M} ) e^{-\eta(t-t_0)}\leq M e^{-\eta(t-t_0)}
\end{equation}
for all $t\geq t_0$, where $M=\alpha_2(\|\varphi(0)\|)+\alpha_3(\|\varphi\|_{\tau})+\bar{M}.$

From condition (i) of Theorem~\ref{Th.ISS}, we have
\begin{equation}\label{bdd}
\|x(t)\|\leq \alpha^{-1}_1\big(M e^{-\eta(t-t_0)}\big)\leq \alpha^{-1}_1\big(M \big)~~\textrm{for all}~t\geq t_0
\end{equation}
and
\begin{equation}
\lim_{t\rightarrow\infty} \|x(t)\| \leq \lim_{t\rightarrow\infty} \alpha^{-1}_1\big(M e^{-\eta(t-t_0)}\big)=0.
\end{equation}
Therefore, system~\eqref{CL.sys} is GUB, that is, $\|x(t)\|\leq \alpha^{-1}_1(M)$ for all $t\geq t_0$, and the trivial solution is GA. The Lyapunov functional candidate $v(t)$ has an exponential convergence rate $\eta$ defined in~\eqref{bound}. Furthermore, if $a=0$, then $\bar{M}=0$ and $M$ depends only on the initial function $\varphi$. Hence, stability of system~\eqref{CL.sys} follows from~\eqref{bdd}.

In the following, Zeno behavior will be excluded from system~\eqref{CL.sys} when $a>0$ and $b<c$. We will do this through identifying a lower bound $T^*>0$ of the inter-execution times $\{t_{i+1}-t_i\}_{i\in\mathbb{N}}$.

Let $L_1$ be the Lipschitz constant of $\alpha^{-1}_1$ on the interval $[0,M]$. Let $R:=\alpha^{-1}_1(M)$ and $L_2$ be the Lipschitz constant of the function $f(t,\cdot,u):\mathcal{C}_{\tau}\rightarrow \mathbb{R}^n$ on $\mathcal{B}^{\tau}_{R}(0)$ for all $(t,u)\in \mathbb{R}^+\times\mathbb{R}^m$. Since $k$ is locally Lipschitz and $f$ is locally Lipschitz with respect to its third argument, the composite function $f(t,\phi,k(\cdot)):\mathbb{R}^n\rightarrow \mathbb{R}^n$ is also locally Lipschitz and we denote $L_3$ as its Lipschitz constant on $\mathcal{B}_{R}(0)$ .

For $t\in [t_i,t_{i+1})$, we have
\[
\dot{\epsilon}(t)=-\dot{x}(t)=-f(t,x_t,k(x(t_i))).
\]
Integrating both sides of the above equation from $t_i$ to $t$ gives
\[
\epsilon(t)-\epsilon(t_i)=-\int^{t}_{t_i} f(s,x_s,k(x(t_i)))\mathrm{d}s.
\]
Since $\epsilon(t_i)=0$, we have
\begin{align}\label{error.estimation}
         &  \|\epsilon(t)\| =    \left\|\int^{t}_{t_i} f(s,x_s,k(x(t_i)))\mathrm{d}s\right\|\cr
         &  =   \int^{t}_{t_i} \left\|f(s,x_s,k(x(t_i)))-f(s,0,k(x(t_i)))+f(s,0,k(x(t_i)))\right\|\mathrm{d}s \cr
         &\leq \int^{t}_{t_i}  L_2 \|x_s\|_{\tau} +L_3 \|x(t_i)\|   \mathrm{d}s \cr
         &\leq \int^{t}_{t_i} \left[ L_2 \alpha^{-1}_1\left(M e^{-\eta(s-\tau-t_0)}\right) +L_3 \alpha^{-1}_1\left(M e^{-\eta(t_i-t_0)}\right)   \right] \mathrm{d}s \cr
         &\leq L_1 L_2 M e^{\eta\tau} \int^{t}_{t_i} e^{-\eta(s-t_0)} \mathrm{d}s + L_1L_3 M (t-t_i)e^{-\eta(t_i-t_0)}\cr
         &=    \lambda_1 (e^{-\eta(t_i-t_0)} - e^{-\eta(t-t_0)})  + \lambda_2 (t-t_i)e^{-\eta(t_i-t_0)}
\end{align}
where $\lambda_1={L_1L_2Me^{\eta\tau}}/{\eta}>0$ and $\lambda_2=L_1L_3 M>0$.

According to our execution rule~\eqref{event}, we have at $t=t_{i+1}$:
\begin{align}
& \sigma \alpha_1(\|x(t_{i+1})\|) + \chi\left(a e^{-b(t_{i+1}-t_0)}\right)\cr
 &=    \chi\left(\|\epsilon(t^-_{i+1})\|\right)\cr
&\leq \chi\Big(\lambda_1 (e^{-\eta(t_i-t_0)} - e^{-\eta(t_{i+1}-t_0)})  + \lambda_2 (t_{i+1}-t_i)e^{-\eta(t_i-t_0)}\Big)
\end{align}
which implies
\[
a e^{-b(t_{i+1}-t_0)} \leq \lambda_1 (e^{-b(t_i-t_0)} - e^{-b(t_{i+1}-t_0)} )  + \lambda_2 (t_{i+1}-t_i)e^{-b(t_i-t_0)}
\]
where $\epsilon(t^-_{i+1})$ represents the left-hand limit of $\epsilon$ at $t=t_{i+1}$, and we replaced $\eta$ with $b$ since $b<c$. Denote the inter-execution time $T_{i+1}=t_{i+1}-t_i$ and multiply both sides of the above inequality with $e^{b(t_i-t_0)}$, then we get
\begin{equation}\label{keyineq}
a e^{-b T_{i+1}} \leq \lambda_1(1-e^{-b T_{i+1}}) +\lambda_2 T_{i+1}.
\end{equation}
Define a function
\[
g(T)=a e^{-b T} - \lambda_1(1-e^{-b T}) -\lambda_2 T~ \textrm{ for } T\geq 0,
\]
then we have $g(0)=a>0$ and $g(T)\rightarrow-\infty$ as $T\rightarrow\infty$. Moreover, $g'(T)=-b(a+\lambda_1)e^{-bT} - \lambda_2<0$ means function $g$ is strictly decreasing on $[0,\infty)$, and there exists a unique $T^*>0$ so that $g(T^*)=0$. Therefore, we can derive from~\eqref{keyineq} that $T_{i+1}\geq T^*$ for all $i\in\mathbb{N}$.
We then can conclude that system~\eqref{CL.sys} does not exhibit Zeno behavior.
\end{proof}

\begin{remark} 
Zeno behavior has been excluded by specifying a uniform minimum time between any two successive events. Such a lower bound can be obtained by solving the equation $g(T)=0$, which has a unique solution (See Example~\ref{example01} in Section~\ref{Sec4} for a demonstration). It can be seen that inequality~\eqref{keyineq} plays a significant role in ensuring the lower bound, and this type of inequalities has been used to exclude Zeno behavior in linear control systems (see, e.g., \cite{WZ-ZPJ:2015}). In this study, we have shown that inequality~\eqref{keyineq} can also be applied to rule out Zeno behavior in nonlinear time-delay control systems with event-triggering condition~\eqref{event} and carefully tuned parameters $\sigma$ and $b$.  
\end{remark}

The advantage of our execution rule is that any positive parameter $a$ will enable to exclude Zeno behavior from the closed-loop system~\eqref{CL.sys}. The drawback of parameter $a$ being positive is that we can show the boundedness and attractivity of the trajectories but the stability of the closed-loop system~\eqref{CL.sys} cannot be derived. The bound $\alpha^{-1}_1(M)$ of the states is not only dependent on the initial data $\varphi$ but is also closely related to the tunable design parameters $a$ and $b$ in~\eqref{et.time}. For example, setting $a$ small and $b$ large decreases the bound at the cost of more events being triggered. On the other hand, setting $a$ large and $b$ small reduces the number of control updates at the cost of increasing the bound. The exponential convergence rate of the Lyapunov candidate depends on parameter $b$, and small $b$ corresponds to small rate $\eta$. Moreover, the lower bound $T^*$ of the inter-execution times also depends on these parameters. Increasing $a$ or reducing $b$ increases this bound and potentially decreases the number of event times over a finite time period (see Table~\ref{table} for a demonstration). 

It can be seen from the above proof that parameter $a$ is essential in ruling out Zeno behavior from system~\eqref{CL.sys}. If $a=0$ with $\sigma>0$, our execution rule reduces to the following one:
\begin{equation}\label{event0}
\chi(\|\epsilon\|) = \sigma \alpha_1(\|x\|)
\end{equation}
which is identical to the one proposed in \cite{KZ-BG:2019}. Although the asymptotic stability of system~\eqref{CL.sys} can be guaranteed from~\eqref{bdd} as discussed in the proof of Theorem \ref{Th}, Zeno behavior cannot be excluded for some linear scalar time-delay systems (see \cite{KZ-BG:2019} for details and Fig. \ref{ETC01} for a demonstration). We can also see from the proof of Theorem \ref{Th} that, if $b<c$, the convergence rate of the Lyapunov functional $V(t,x_t)$ is $\eta=b$. Compared to system~\eqref{CL.sys} associated with the execution rule~\eqref{event0}, the convergence rate is smaller and less control updates are triggered, which implies that it might be possible to exclude Zeno behavior from system~\eqref{CL.sys} and Theorem~\ref{Th} confirms this possibility. If $\sigma=0$ with $a>0$, then we derive from~\eqref{event} the following execution rule
\begin{equation}\label{event1}
\|\epsilon(t)\|=a e^{-b(t-t_0)}.
\end{equation}
We can conclude from Theorem \ref{Th} that system~\eqref{CL.sys} with the event times determined according to the execution rule~\eqref{event1} does not exhibit Zeno behavior provided $0<b<\mu$. Compared with the execution rule~\eqref{event} which depends on both the states and the evolution time, the execution rule~\eqref{event1} is only time-dependent and thus simpler to implement, but more control updates are most likely triggered (see Table \ref{table2} for numerical comparisons).

Quadratic forms as Lyapunov functions have been widely used in stability analysis of control systems, e.g., $V_1(t,x)=x^TPx$ where $P$ is a positive definite $n\times n$ matrix. For this type of Lyapunov functions, condition (i) of Theorem \ref{Th.ISS} holds with $\alpha_1(s)=\gamma s^2$ where $\gamma>0$ is some constant (e.g., $\gamma$ is the smallest eigenvalue of $P$ if $V_1(t,x)=x^TPx$), and then $\alpha^{-1}_1(s)=\sqrt{s/\gamma}$, which is locally Lipschitz on the open interval $(0,\infty)$ but not at $s=0$. Therefore, inequality~\eqref{error.estimation} is not valid, and the technique of ruling out Zeno behavior in Theorem~\ref{Th} is not applicable to such class $\mathcal{K}_{\infty}$ functions. In the next theorem, we will show the lack of Zeno behavior from system~\eqref{CL.sys} by a direct use of Definition~\ref{Zeno}, if $\alpha^{-1}_1$ is not locally Lipschitz at zero. Nevertheless, the lower bound may not be guaranteed.

\begin{theorem}\label{Th2}
Suppose all the conditions of Theorem \ref{Th} hold with the Lipschitz condition of $\alpha^{-1}_1$ replaced with the following one:
\begin{itemize}
\item $\alpha^{-1}_1$ is locally Lipschitz on the open interval $(0,\infty)$ but not at zero.
\end{itemize}
Then,
\begin{itemize}
\item[1)] if $a=0$ and $\sigma>0$, the closed-loop system~\eqref{CL.sys} is GAS;

\item[2)] if $a>0$, the closed-loop system~\eqref{CL.sys} is GUB and its trivial solution is GA. Furthermore, if $b<\mu-\sigma$, system~\eqref{CL.sys} does not exhibit Zeno behavior.
\end{itemize}
\end{theorem}
\begin{proof} The GAS, GUB, and GA properties of system~\eqref{CL.sys} follow directly from the proof of Theorem~\ref{Th}. We only need to show that system~\eqref{CL.sys} does not exhibit Zeno behavior when $a>0$ and $b<\mu-\sigma$. We will do this by a contradiction argument. Assume there exists an initial condition so that the corresponding solution of system~\eqref{CL.sys} exhibits Zeno behavior, that is, the sequence of event times $\{t_i\}_{i\in\mathbb{N}}$ is upper bounded by a finite number. Hence, the sequence $\{t_i\}_{i\in \mathbb{N}}$ is convergent, and we denote $t^*:=\lim_{i\rightarrow\infty} t_i$ where $t^*> t_0$. Since $\alpha^{-1}_1$ is locally Lipschitz on $(0,\infty)$, let $\bar{L}_1$ represent the Lipschitz constant of $\alpha^{-1}_1$ on the closed interval $[M e^{-\eta(t^*-t_0)}, M]$ with $M$ and $\eta$ defined in the proof of Theorem \ref{Th}. For any small $\varepsilon>0$, there exists a $t_i$ such that $t^*-\varepsilon\leq t_i< t^*$, and we then can obtain from~\eqref{error.estimation} that
\[
\|\epsilon(t)\| \leq   \bar{\lambda}_1 (e^{-\eta(t_i-t_0)} - e^{-\eta(t-t_0)})  + \bar{\lambda}_2 (t-t_i)e^{-\eta(t_i-t_0)}
\]
for all $t\in [t_i,t_{i+1})$, where $\bar{\lambda}_1={\bar{L}_1L_2Me^{\eta\tau}}/{\eta}>0$ and $\bar{\lambda}_2=\bar{L}_1L_3 M>0$ with $L_2$ and $L_3$ given in the proof of Theorem \ref{Th}. Similarly to the discussion of~\eqref{keyineq}, we can derive a lower bound of $t_{i+1}-t_i$, that is, there exists a $\bar{T}>0$ so that $t_{i+1}-t_i\geq \bar{T}$. Since $\varepsilon>0$ can be chosen arbitrarily small, we let $\varepsilon<\bar{T}$ which then implies $t_{i+1}\geq t_i+\bar{T}\geq t^* -\varepsilon +\bar{T}> t^*.$ This gives a contradiction to the definition of $t^*$. Therefore, system~\eqref{CL.sys} does not exhibit Zeno behavior.
\end{proof}


\section{Illustrative Examples}\label{Sec4}


\begin{example}\label{example01} Consider a linear scalar control system with time delay
\begin{eqnarray}\label{linear.sys}
\left\{\begin{array}{ll}
\dot{x}(t)=B x(t-r)+u(t) \cr
x_{t_0}=\varphi
\end{array}\right.
\end{eqnarray}
where state $x\in \mathbb{R}$, initial function $\varphi(s)=1$ for $s\in [-r,0]$, time delay $r=16$, $B=-0.1$, and $u(t)=kx(t)$ is the state feedback control with control gain $k=-0.2$. It has been shown in \cite{KZ-BG:2019} that system~\eqref{linear.sys} with $u=0$ is unstable, and we can conclude from Theorem \ref{Th.ISS} that the feedback control system~\eqref{linear.sys} is GAS.
\end{example}

Consider the event-triggered implementation of $u$ in system \eqref{linear.sys}, and then the closed-loop system can be written in the form of \eqref{CL.sys}:
\begin{eqnarray}\label{linear.CLsys}
\left\{\begin{array}{ll}
\dot{x}(t)=B x(t-r)+kx(t)+k\epsilon(t) \cr
x_{t_0}=\varphi
\end{array}\right.
\end{eqnarray}
where $\epsilon(t)=x(t_i)-x(t)$ for $t\in [t_i,t_{i+1})$ with $i\in \mathbb{N}$, and the sequence of event times $\{t_i\}_{i\in\mathbb{N}}$ is to be determined by \eqref{et.time}.

To derive the functions $\alpha_1$ and $\chi$ in \eqref{et.time} and illustrate the effectiveness of Theorem~\ref{Th}, we consider Lyapunov functional $V=V_1+V_2$ with 
\[
V_1(\phi(0))=p|\phi(0)| \textrm{ and } V_2(\phi)=\int^0_{-r} |\phi(s)| \left(\frac{-s}{r}q_1 + \frac{s+r}{r}q_2 \right) \mathrm{d}s,
\]
where constants $p$, $q_1$, and $q_2$ are positive and satisfy $-k>q_2>q_1=p|B|$. Then condition (i) of Theorem \ref{Th.ISS} is satisfied with $\alpha_1(|\phi(0)|)=\alpha_2(|\phi(0)|)=p|\phi(0)|$, while condition (ii) holds with 
$\alpha_3(\|\phi\|_r)= r q_2~\|\phi\|_r$. It can be seen that $\alpha^{-1}_1(z)=z/p$ is globally Lipschitz on its domain.

From the dynamics of system \eqref{linear.CLsys}, we have
\begin{align*}
\mathrm{D}^+{V}_1(x(t))  =  \mathrm{D}^+ |x(t)| &\leq p |B| |x(t-r)|+ p k \mathrm{sgn}(x)x+ p |k \epsilon|\cr
          &  =  |B| V_1(x(t-r))+ k V_1(x(t))+ p |k \epsilon|
\end{align*}
where $\mathrm{sgn}(x)$ represents the sign function of the real number $x$, and
\begin{align*}
\mathrm{D}^+{V}_2(x_t) &=  \frac{\mathrm{d}}{\mathrm{d}t} \left(\int^t_{t-r} |x(s)| \left(\frac{t-s}{r}q_1 + \frac{s-t+r}{r}q_2 \right) \mathrm{d}s  \right)\cr
 & = -\frac{q_2-q_1}{r} \int^t_{t-r} |x(s)| \mathrm{d}s  +q_2 |x(t)| - q_1|x(t-r)|\cr
 &\leq -\frac{q_2-q_1}{rq_2} V_2(x_t)  + \frac{q_2}{p} V_1(x(t)) - \frac{q_1}{p}V_1(x(t-r)).
\end{align*}
The above two inequalities imply
\begin{align*}
\mathrm{D}^+{V} (x_t) \leq (k+q_2)V_1(x(t)) - \frac{q_2-q_1}{rq_2} V_2(x_t) +p |k \epsilon|.
\end{align*}
Then condition (iii) of Theorem~\ref{Th.ISS} is satisfied with
\[
\mu=\min\left\{-(k+q_2), \frac{q_2-q_1}{rq_2}\right\}>0 \textrm{ and } \chi(|\epsilon|)= p|k\epsilon|
\]
and the execution rule~\eqref{event} can be written as follows
\begin{equation}\label{exrule00}
|\epsilon|=\frac{\sigma}{|k|} |x|+a e^{-b(t-t_0)}.
\end{equation}
Thus, all the conditions of Theorem~\ref{Th} are satisfied. From the definitions of $\mu$ and $q_1$, we can see that  choosing small enough $p$ and $q_2$ while keeping $p/q_2$ close to $0$ can make $\mu$ arbitrarily close to $\min\{-k,1/r\}=0.0625$. In the following simulations, we select small $p$ and $q_2$ so that $p/q_2=0.08$, and then $\mu=0.062$.

To illustrate our results with numerical simulations, we choose $\sigma=0.03$ so that $\mu>\sigma$. Figure~\ref{ETC01} shows the trajectory of system~\eqref{linear.CLsys} with $a=0$ in~\eqref{exrule00} and numerically verifies the existence of Zeno behavior (theoretical proof of the existence can be found in \cite{KZ-BG:2019}). Let $a=0.3$ and $b=0.03$, then $\mu-\sigma>b$ and Fig.~\ref{ETC02} shows the trajectories of system~\eqref{linear.CLsys} with execution rule~\eqref{exrule00}. For system~\eqref{linear.CLsys} with the above Lyapunov functional, all the functions mentioned in Theorem~\ref{Th} are linear, and the corresponding Lipschitz constants can be derived easily. It can be obtained from~\eqref{keyineq} with~\eqref{exrule00} that the lower bound of the inter-execution times $\{t_{i+1}-t_i\}_{i\in\mathbb{N}}$ is $3.59\times 10^{-3}$ which is the unique solution of the equation $g(T)=0$ given in the proof of Theorem~\ref{Th}. This equation can be readily solved by Matlab or Maple.

Table~\ref{table} indicates the number of events triggered by the execution rule~\eqref{exrule00} over the time interval $[0,10^3]$ for different combinations of the parameters $a$ and $b$. It can be seen that less events will be triggered for bigger $a$ and/or smaller $b$, which verifies our discussion in Section~\ref{Sec3}. Table~\ref{table2} shows the number of event times determined by \eqref{exrule00} on the interval $[0,600]$ for different values of $a$ with $b=0.05$. We can observe that more control updates are triggered according to the time-dependent execution rule~\eqref{event1} (the execution rule~\eqref{exrule00} with $\sigma=0$). 

\begin{table}[!t] 
\caption{Number of the Event Times Determined by~\eqref{exrule00} with $\sigma=0.03$ on the Time Interval $[0,10^3]$.}
\label{table}
\centering
\scalebox{1.}{
\begin{tabular}{ |c|c|c| }
\hline
$a$ & $b$ & Number of event times \\ 
\hline
\multirow{2}{1.5em}{$0.1$} 
& 0.01 & 117 \\ \cline{2-3}
& 0.03 & 174 \\ \cline{2-3}
\hline
\multirow{2}{1.5em}{$0.3$} 
& 0.01 & 89 \\ \cline{2-3}
& 0.03 & 163 \\ \cline{2-3}
\hline 
\end{tabular} }
\end{table}

\begin{table}[!t]
\caption{Number of the Event Times Determined by~\eqref{exrule00} on the Time Interval $[0,600]$.}
\label{table2}
\centering
\scalebox{1.}{
\begin{tabular}{ |c|c|c|c|c| } 
\hline
$a~~(b=0.05)$ & 0.1 & 0.4 & 1 & 2 \\ \cline{1-5}
Number of event times with $\sigma=0.01$ & 1160 & 699 & 518 & 492 \\ \cline{1-5}
Number of event times with $\sigma=0$ & 1892 & 854 & 604 & 548 \\ \cline{1-5}
\hline
\end{tabular} }
\end{table}

\begin{figure}[!t]
\centering
\subfigure[$a=0$]{\label{ETC01}\includegraphics[width=2.4in]{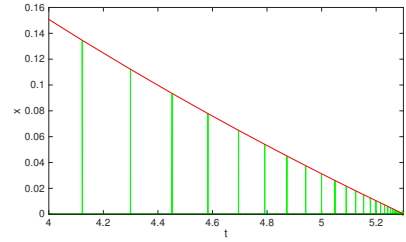}}
\subfigure[$a=0.3$ and $b=0.03$]{\label{ETC02}\includegraphics[width=2.4in]{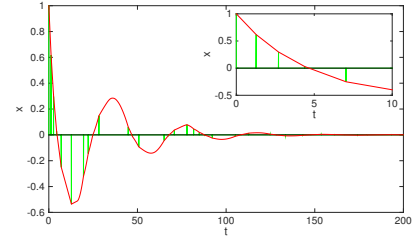}}
\caption{Trajectories of the closed-loop system \eqref{linear.CLsys} with event times determined by \eqref{exrule00}. }
\label{fig1}
\end{figure}

\begin{example}\label{example02} Consider the following nonlinear time-delay control system
\begin{eqnarray}\label{nonlinear.sys}
\dot{x}(t)=A x(t)+g(x(t-d))+Bu(t)
\end{eqnarray}
where state $x\in \mathbb{R}^n$, time delay $d>0$, matrices $A\in\mathbb{R}^{n\times n}$ and $B\in\mathbb{R}^{n\times m}$, and $u(t)=Kx(t)$ is the state feedback control with control gain $K\in \mathbb{R}^{m\times n}$. For any $x,y\in\mathbb{R}^n$, nonlinear function $g\in\mathcal{PC}(\mathbb{R}^n,\mathbb{R}^n)$ satisfies $\|g(x)-g(y)\|\leq L\|x-y\|$ with $L>0$.
\end{example}

With the event-triggering implementation and measurement error $\epsilon(t)=x(t_i)-x(t)$ for $t\in[t_i,t_{i+1})$, we can rewrite system \eqref{nonlinear.sys} as follows
\begin{eqnarray}\label{nonlinear.CLsys}
\dot{x}(t)=(A+BK) x(t)+g(x(t-d))+BK\epsilon(t)
\end{eqnarray}
where the event times $t_i$ ($i\in\mathbb{N}$) are to be determined by event-triggering condition \eqref{event}. To do so, let us apply Theorem \ref{Th2} with the Lyapunov functional candidate $V=V_1+V_2$ where 
\begin{align*}
&V_1(\phi(0))=\phi^T(0)\phi(0),\cr
&V_2(\phi)=\int^0_{-d} \phi^T(s)\left( \frac{d+s}{d}q+1 \right) \phi(s) \mathrm{d}s.
\end{align*}
Here, $q$ is a positive constant. Then, both conditions (i) and (ii) of Theorem \ref{Th} hold and $\alpha_1(\|\phi\|)=\|\phi\|^2$ which is not locally Lipschitz at zero. It follows from the dynamics of system~\eqref{nonlinear.CLsys} that
\begin{align*}
\textrm{D}^+ V_1 (x(t))\leq& \left( \Lambda +\varepsilon_1+\varepsilon_2 \right)V_1(x(t))\cr
                         &+\varepsilon^{-1}_1 L^2 V_1(x(t-d)) + \varepsilon^{-1}_2 \epsilon^T(t)(BK)^TBK\epsilon^T(t)
\end{align*}
and
\begin{align*}
\textrm{D}^+ V_2 (x_t)\leq& (q+1)V_1(x(t)) - V_1(x(t-d)) - \frac{q}{d}\int^t_{t-d} x^T(s)x(s)\mathrm{d}s
\end{align*}
where $\Lambda$ denotes the largest eigenvalue of the matrix $(A+BK)^T(A+BK)$. In the first inequality, we applied Young's inequality twice with $\varepsilon_1=L^2$ and $\varepsilon_2>0$. We get from these two inequalities that condition (iii) of Theorem~\ref{Th} holds with $\chi(\|\epsilon\|) = \varepsilon^{-1}_2 \|BK\|^2 \|\epsilon\|^2$ and 
\begin{align*}
\mu=\min\left\{  -\left(\Lambda +L^2+1+q+\varepsilon_2\right), ~\frac{q}{d(q+1)} \right\}.
\end{align*}
Then, the sequence of event times defined by~\eqref{et.time} is as follows
\begin{align}\label{et.timeL}
t_{i+1}=\inf\left\{t>t_i \mid \|\epsilon\|^2= \frac{\sigma\varepsilon_2}{\|BK\|^2} \|x\|^2 + a^2 e^{-2b(t-t_0)}    \right\}.
\end{align}
Therefore, we can conclude from Theorem~\ref{Th2} that if $\mu>0$ and we select positive parameters $\sigma$ and $b$ in~\eqref{et.timeL} so that $b<\mu-\sigma$, then closed-loop system~\eqref{nonlinear.CLsys} with event times determined by~\eqref{et.timeL} is free of Zeno behavior and GUB, and its trivial solution is GA. To ensure a positive $\mu$, it is sufficient to require 
$$\Lambda +L^2+1<0,$$ 
and then positive constants $\varepsilon_2$ and $q$ can be chosen close to zero so that $\Lambda +L^2+1 +q + \varepsilon_2<0$.

\section{Conclusions}\label{Sec5}
We have investigated event-triggered control for nonlinear time-delay systems. A novel execution rule has been proposed with tunable parameters. Analysis of many different combinations of these parameters has been conducted. Numerical simulations have been provided to verify our theoretical results. We can see from the proof of Theorem~\ref{Th} that if $b\geq c$, the exponential convergence rate of the Lyapunov functional candidate is $\eta=c$ which is the same as the convergence rate for system~\eqref{CL.sys} with execution rule~\eqref{event0}. Since system~\eqref{CL.sys} with execution rule~\eqref{event0} can exhibit Zeno behavior, intuitively we expect that it is still possible for system~\eqref{CL.sys} with execution rule~\eqref{event} to exhibit Zeno behavior when $b\geq c$. Discussion in this direction will be explored further which is a topic for future research. In this study, we considered time delays in the nonlinear systems but assumed that the feedback controller receives the system states and updates the control signals simultaneously. Therefore, another future research topic is to consider time delays in the implementation of the proposed event-triggered control algorithm and also in the feedback controllers.

\end{document}